\documentclass[11pt]{article}
\usepackage{graphicx} 

\usepackage{amsmath}
\usepackage{amssymb}

\usepackage{pre}
\usepackage[colorinlistoftodos,prependcaption,textsize=tiny]{todonotes}

\def\ANON{0} 
\def\COMM{1} 

\title{$\QACZ$ Can Prepare Every Logarithmic-Qubit State}
\ifnum\ANON=1
\author{
  Anonymous Authors.
}
\else
\author{
  Lucas Gretta \thanks{University of California at Berkeley. \ Email: \url{lucas_gretta@berkeley.edu}. \ Supported by NSF Award CCF-2231095.}
  \and 
  Meghal Gupta\thanks{University of California at Berkeley. \ Email: \url{meghal@berkeley.edu}. \ Supported by NSF GRFP.}
  \and
  Malvika Raj Joshi\thanks{University of California at Berkeley. \ Email: \url{malvika@berkeley.edu}. \ Supported by NSF grant 2311733 and DOE grant DE-SC0024124.}}
\fi
\date{}

\ifnum\COMM=1
\newcommand{\malvika}[1]{\textcolor{purple}{[\textbf{Malvika:} {#1}]}}
\newcommand{\luke}[1]{\textcolor{blue}{[\textbf{Luke:} {#1}]}}
\newcommand{\meghal}[1]{\textcolor{red}{[\textbf{Meghal:} {#1}]}}
\newcommand{\gpt}[1]{\textcolor{orange}{[\textbf{GPT:} {#1}]}}
\newcommand{\todomal}[1]{\todo[linecolor=Plum,backgroundcolor=Plum!25,bordercolor=Plum]{\textbf{@mal todo}: #1}}
\newcommand{\todoluke}[1]{\todo[linecolor=blue,backgroundcolor=blue!25,bordercolor=blue]{\textbf{@luke todo}:#1}}
\newcommand{\todomeghal}[1]{\todo[linecolor=red,backgroundcolor=red!25,bordercolor=red]{\textbf{@meghal todo} #1}}
\newcommand{\todoany}[1]{\todo[linecolor=orange,backgroundcolor=orange!25,bordercolor=orange]{\textbf{todo @plz do}: #1}}
\else
\newcommand{\malvika}[1]{}
\newcommand{\luke}[1]{}
\newcommand{\meghal}[1]{}
\newcommand{\gpt}[1]{}
\newcommand{\todomal}[1]{}
\newcommand{\todoluke}[1]{}
\newcommand{\todomeghal}[1]{}
\newcommand{\todoany}[1]{}
\fi

\begin{document}

\maketitle
\ifnum\ANON=1
\PackageWarningNoLine{Global}{Note anonymous mode}
\ifnum\COMM=1
\textcolor{red}{\textbf{Warning:}{ Comments are enabled in anonymous mode}}
\fi
\fi
\ifnum\COMM=1
\PackageWarningNoLine{Global}{Note comments are enabled}
\fi

\begin{abstract}
$\QACZ$ is the class of constant-depth $\poly(n)$-ancilla circuits obtained by extending $\QNC^0$, the class of local circuits, to include nonlocal interactions via arbitrary width Toffoli gates. It is believed to be weaker than its counterpart, $\QNC^0_f$, obtained by including arbitrary-width $\FANOUT$ gates instead ($\QACZ \subseteq \QNC^0_f$ \cite{moore1999qac0}). 

In this note, we show that every $O(\log n)$-qubit state can be exactly and cleanly prepared by a  $\poly(n)$-ancilla $\QACZ$ circuit. Previous known $\poly(n)$-ancilla circuits for arbitrary such states are only known via additional access to either $\FANOUT$ or $\QRAM$ (indexing) gates \cite{rosenthal2021query, gretta2026polylogarithmicweightdickestatesqac0}, neither of which are known to be in $\QACZ$. Equivalently, prior constructions of arbitrary $n$-qubit states in $\QACZ$ require \emph{doubly exponential} size and we obtain an exponential factor improvement.   
\end{abstract}

\section{Introduction}
Arbitrary state preparation necessarily requires nonlocal operations well beyond $\QNC^0$, the class of constant-depth local quantum circuits \cite{nielsen2000quantum, moore2001parallel}.  Formally, $\QNC^0$ circuit consists of  $O(1)$ number of layers of local (e.g. two-qubit) gates. One natural extension of $\QNC^0$ is obtained through the inclusion of $\FANOUT$ gates to enable global entangling operations. The resulting class $\QNC^0_f$ is extremely powerful and can, for instance, simulate classical $\TC^0$ circuits, perform the Quantum Fourier Transform and prepare arbitrary symmetric states \cite{hoyer2005fanout, takahashi2012collapse, gretta2026polylogarithmicweightdickestatesqac0}. Therefore, it is a natural to wonder whether $\FANOUT$ can be replaced by potentially weaker nonlocal interactions such as multi-qubit Toffoli gates. $\QACZ$ is the class of constant-depth circuits formed by multi-qubit Toffoli gates and local gates \cite{moore1999qac0}.

From a fault tolerance perspective, $\FANOUT$ may appear to be a simpler gate than Toffoli, due to it being Clifford \cite{gottesman1998heisenbergrepresentationquantumcomputers}. However, Toffoli gates can be efficiently simulated using only $\FANOUT$ and local gates, i.e. $\QACZ \subseteq \QNC^0_f = \QACZF$ \cite{hoyer2005fanout, takahashi2012collapse}. 
On the other hand,  whether $\QACZ$ can implement the $\FANOUT$ operation, or equivalently the $\PARITY$ function, is a longstanding open question \cite{moore1999qac0}. It is unknown whether $\QACZ$ can simulate all decision problems in $\ACZ$, its classical analogue, let alone compute $\PARITY$ which is very much not in $\ACZ$ \cite{hastad1986switch, lmn1993ac0}. Recently, Grier, Morris and Wu showed that $\QACZ$ can compute any \emph{symmetric function} in $\ACZ$, but their result does not apply to other natural $\ACZ$ functions such as the indexing function \cite{grier2026tc0}.  

To our knowledge, the first construction for arbitrary state preparation in this setting is due to Rosenthal in 2021 \cite{rosenthal2021query}. Rosenthal showed that every $O(\log n)$-qubit state is preparable by a $\poly(n)$-ancilla $\QACZF$ circuit. Recently, \cite{gretta2026polylogarithmicweightdickestatesqac0} showed that $\poly(n)$-ancilla $\QACZ[\QRAM]$ circuits suffice instead. These are $\QACZ$ circuits augmented with $\QRAM$, a gate that implements the indexing function. $\QACZ[\QRAM]$ can be simulated by $\QACZF$ but the converse is unknown, making $\QRAM$ a potentially weaker resource than $\FANOUT$. In this note, we describe an extension to \cite{gretta2026polylogarithmicweightdickestatesqac0} that eliminates the dependence on $\QRAM$. Specifically, we show that $O(\log n)$-qubit states can be prepared with $\poly(n)$ ancillae in plain $\QACZ$. 

We emphasize that it is remains unknown whether $\QRAM$ itself can be implemented in $\QACZ$, or whether either containment in $\QACZ \subseteq \QACZ[\QRAM] \subseteq \QACZF$ is strict.

\begin{theorem}[Arbitrary state synthesis in $\QAC$]\label{thm:anylogn}
For every $c \log n$-qubit state $\ket{\psi}$, there exists a $c_1$-depth, $n^{c_2}$-ancilla $\QACZ$ circuit that exact and cleanly prepares $\ket{\psi}$, where $c_1, c_2$ are constants that only dependent on $c$.
\end{theorem}
\noindent Note that naive preparations of arbitrary states in constant-depth are \emph{doubly-exponential} in the number of target qubits (thus $\exp(n)$ for $O(\log n)$ qubits). \cref{thm:anylogn} implies arbitrary state preparation in $\QACZ$ using exponential ancillae in the number of target qubits. We provide the proof in the following section \ref{proof:here}.

\section{Arbitrary state preparations in \texorpdfstring{$\QACZ$}{QAC0}}
The proof heavily uses the state-preparation machinery introduced in \cite{gretta2026polylogarithmicweightdickestatesqac0} and follows the same notations and conventions. We also assume familiarity with standard results about $\QACZ$, and refer the reader to the preliminaries in \cite{gretta2026polylogarithmicweightdickestatesqac0}. 

For $i \in [n]$, let $e_i$ denote the $n$ bit string $0^{i-1} 1 0^{n-i-1}$.
It was shown in Lemma 4.1 of \cite{gretta2026polylogarithmicweightdickestatesqac0} that one can prepare arbitrary superpositions over the $\ket{e_i}$ in $\QACZ$. 
One can also implement the map $\ket{e_i} \ket{0} \mapsto \ket{e_i} \ket{i}$ in $\QACZ$, since it has a $\log n$-fanout classical circuit. However, the key challenge lies in uncomputing the one-hot register, i.e. the map $\ket{e_i} \ket{i} \mapsto \ket{0^n} \ket{i}$, which is equivalent to $\QRAM$. As we show below, this map is not necessary for the construction. 

\paragraph{Proof Outline.}  First, we weakly approximate  the target state $\ket{\psi_*}$ by preparing a state with a marked component on $\ket{\psi_*}$ of at least $1/\poly(n)$ amplitude in \cref{cl:smampl}. Next, we present a stronger version of the ``parallel amplification'' technique, to amplify this amplitude without relying on $\FANOUT$. The idea of parallel amplification was first introduced in \cite{rosenthal2021qac0}, and has since appeared in various forms in subsequent works \cite{cleandicke2026, gretta2026paritynotinqac0iff, gretta2026polylogarithmicweightdickestatesqac0}. In \cref{lem:ctrswap}, we show that amplifying the success probability of a marked state requires $\FAN$ proportional to the \emph{number of target qubits} rather than the number of parallel runs. In particular, this $1/\poly(n)$ weak approximation of $\ket{\psi_*}$ can be amplified with $\FAN_{\log n}$, even though it requires $\poly(n)$ parallel preparations. 

Finally, we clean up the ancillae using a general approach, stated in \cref{cor:clean}. 

\begin{claim}\label{cl:smampl}
There exist constants $c_0,c_1$ such that, for every $\ell$-qubit state $\ket{\psi_*} = \sum_{i \in [n]} \gamma_i \ket{i}$ where $n = 2^{\ell}$ and $\gamma_i \in \mathbb{C}$ are arbitrary normalized amplitudes, there is a $c_0$-depth, $n^{c_1}$-ancilla $\QACZ$ circuit $C_{\varphi}$  to prepare the state, 
$$\ket{\varphi} := \frac{1}{\sqrt{n}} \ket{W_n}  \ket{\psi_*} \ket{1} + \sqrt{1-\frac{1}{n}} \ket{\bad} \ket{0}$$
\end{claim}
\begin{proof}
Begin by preparing the desired superposition on $n$-qubit one-hot vectors, using Lemma 4.1 of \cite{gretta2026polylogarithmicweightdickestatesqac0}. 
Note although this lemma as stated only describes non-negative real amplitudes, one can extend it to arbitrary amplitudes by applying a layer of single-qubit phase gates.  
\begin{align}
\ket{\psi_0} := \sum_{i \in [n]} \gamma_i \ket{e_i}_A 
\end{align} 
Next, using the map $\ket{e_i} \ket{0^{\ell}} \mapsto \ket{e_i} \ket{i}$, which is in $\QACZ$ (Fact 3.12 \cite{gretta2026polylogarithmicweightdickestatesqac0}), obtain,  
\begin{align}
    \ket{\psi_1} &:= \sum_{i \in [n]} \gamma_i \ket{e_i}_A \ket{i}_Q \\
    &= \frac{1}{\sqrt{n}}\ket{W_n}_A   \ket{\psi_*}_Q  + \sqrt{1-\frac{1}{n}} \ket{\nu'}_{A,Q}, \label{eq:herew}
\end{align}
where $\ket{\nu'}$ is an arbitrary state orthogonal to the $\ket{W_n}$ subspace on $A$. 

Since the state $\ket{W_n} \ket{-}$ is cleanly preparable in $\QACZ$, one can also implement the reflection $R := (I - 2\kb{W_n} \otimes \kb{-})$ in $\QACZ$. Applying this on $A$ and a fresh register $t$ gives, 
\begin{align}
    \ket{\psi_2} &:= R(A,t) \ket{\psi_1}_{A,Q} \ket{0}_t \\ 
    &= \frac{1}{\sqrt{n}}\ket{W_n}_A  \ket{\psi_*}_Q \ket{1}_t +  \sqrt{1-\frac{1}{n}} \ket{\nu'}_{A,Q} \ket{0} 
\end{align}
\end{proof}
\noindent To amplify the amplitude on $\ket{\psi_*}$ without blowing up the depth, we will first show that the following restricted controlled-swap primitive from \cite{gretta2026polylogarithmicweightdickestatesqac0} can be implemented in $\QACZ$ without $\FAN_n$.  
\begin{lemma}\label{lem:ctrswap}
For a $n$ qubit registers $A = {a_1, a_2 \dots a_n}$  and  $B = {b_1, b_2 \dots b_n}$ and single qubit $t$, the following map can be cleanly implemented in $\QACZ$. 
$$\ket{e_i}_A \ket{\psi}_{B,t}  \mapsto \ket{e_i}_A \lr{\swap(b_i, t) \otimes I_{B\setminus b_i}} \ket{\psi}_{B,t}$$ 
for all $i \in [n]$, and arbitrary $n+1$ qubit state $\ket{\psi}_{B,t}$ 
\end{lemma}
\noindent Note the action of the map outside the Hamming weight $1$ subspace on $A$ is allowed to be arbitrary.
\begin{proof}
The desired map has the same action as the unitary $U = \prod_{i \in [n]} \cswap(a_i, b_i, t)$ restricted to the subspace $\S := \spn \clr{\ket{0^n}, \ket{e_1}, \ket{e_2} \dots \ket{e_n}}$ on $A$.  
The proof follows from the same decomposition of the restricted map $U_{\S}$ as in \cite{cleandicke2026}. 

Recall $\cswap(a_i, b_i, t) = \ccnot(a_i, b_i, t) \cdot \ccnot(a_i, t, b_i) \cdot \ccnot(a_i, b_i, t)$ and define,
\begin{align}
    U_1 &:= \prod_{i \in [n]} \ccnot(a_i, b_i, t), \\
    U_2 &:= \prod_{i \in [n]} \ccnot(a_i, t, b_i) \\
    &= \prod_{i \in [n]} H_{b_i} \otimes H_{t} \ccnot(a_i, b_i, t) H_{b_i} \otimes H_{t} \\
    &= H^{\otimes n+1}_{B,t} \cdot U_1 \cdot H^{\otimes n+1}_{B,t} .\label{eq:here1}
\end{align}
The $\ccnot$ gates on different $(a_i,b_i)$ commute inside $\S$, therefore $(U_1 \cdot U_2 \cdot U_1)|_{\S} = U|_{\S}$ and due to \cref{eq:here1} it suffices to implement $U_1|_{\S}$ in $\QACZ$.

Observe that since the $\PARITY$ and $\OR$ functions on $A$ are equivalent inside $\S$, 
$U_1|_{\S}$ is equivalent to computing following tribes function on $t$,  
\begin{align}
    f(x, y) = \bigvee_{i \in [n]} (x_i \land y_i)
\end{align}
Formally, for $x \in \clr{e_i}_{i \in [n]}$ and $y \in \bin^n$ and $z \in \bin$, 
\begin{align}
    U_1 \ket{x}_A \ket{y}_B \ket{z}_t &= \ket{x}_A \ket{y}_B \ket{z \oplus f(x,y)}_t 
\end{align}
$f(x,y)$ is a read-once formula and can be cleanly computed in depth-$3$ $\QACZ$ as follows: (1)  compute each $x_i \land y_i$ in parallel onto fresh $n$-qubit ancilla register $Q$ (2) compute the $\OR$ of $Q$ onto $t$, i.e, the gate $X_t \cdot (I - 2\kb{0^n}_Q \otimes \kb{-}_t)$ (3) uncompute $Q$ by running the same gates as (1).

Consequently, $U|_{\S}$ can be implemented by a depth-$9$ $\QACZ$ circuit.
\end{proof}

\noindent This provides a stronger version of Claim 4.2 \cite{gretta2026polylogarithmicweightdickestatesqac0}.
\begin{restatable}[Parallel amplification with limited fanout]{corollary}{parallelamp} \label{cor:parallel}
For any $\ell$-qubit state $\ket{\psi_*}$, suppose there exists a depth $d$, $m$-ancilla $\QACZ$ circuit $C$ to prepare the state 
$$\ket{\varphi} := \sqrt{\alpha}  \ket{\psi_*}_Q \ket{\mu}_A \ket{1} + \sqrt{1-\alpha} \ket{\nu}_{Q,A} \ket{0},$$
for some $\alpha \in (0,1]$ such that $\alpha^{-1} \leq \poly(n)$ and arbitrary states $\ket{\mu}$ and $\ket{\nu}$.
Then there exists a $\QACZ[\FAN_{\ell}]$ circuit $C'$ of depth $d' = O(d)$ using $m' = O((m+n) /\alpha + 1/\alpha^2)$ ancillae that constructs $\ket{\psi_*} \ket{\eta}$ where $\ket{\eta}$ is an arbitrary ancilla register.   
\end{restatable}
\noindent Thus, when $\ell = \polylog(n)$, the above amplification runs in $\QACZ$. 
The proof is exactly the same as Claim 4.7 of \cite{gretta2026paritynotinqac0iff} but uses the stronger \cref{lem:ctrswap} instead. We provide it in \cref{sec:aproofs}. 

The final piece is the following general consequence of Lemma 10 of \cite{gretta2026polylogarithmicweightdickestatesqac0} that additionally provides a \emph{clean preparation} of $\ket{\psi_*}$. 
\begin{corollary}[Dirty-to-Clean state preparation]\label{cor:clean}
For any $\ell$-qubit state $\ket{\psi_*}$, suppose there exists a depth $d$, $n$-ancilla $\QACZ[\FAN_k]$ circuit $C$ to prepare the state,
$$\ket{\varphi} = \frac{\ket{0^\ell} \ket{0}+ \ket{\psi_*} \ket{1}}{\sqrt{2}} \otimes \ket{\nu}$$
where $\ket{\nu}$ may be an arbitrary state. Then, there exists a depth $O(d)$, $O(n + \ell)$-ancilla circuit that \emph{cleanly} prepares $\ket{\psi_*}$. 
\end{corollary}
\begin{proof}
Applying Lemma 10 of \cite{gretta2026polylogarithmicweightdickestatesqac0} to $C$ provides two circuits $C_0, C_1$ that cleanly prepare the states $\ket{\psi_*} \ket{\nu}$ and $\ket{0^{\ell}} \ket{\nu}$ respectively. 

Then, one can cleanly prepare $\ket{\psi_*}$ as follows: (1) start with registers $\ket{0^{\ell}}_X \ket{0^{\ell}}_Y \ket{0^{n}}_A$, (2)  run $C_0$ on $X,A$ and, (3) run $C_1^\dag$ ($C_1$ in reverse) on $Y,A$. 
To verify, 
\begin{align}
    C_1^\dag(Y,A) C_0(X,A) \ket{0^{\ell}}_X \ket{0^{\ell}}_Y \ket{0^{n}}_A &= C_1^\dag(Y,A) \ket{\psi_*}_X \ket{0^{\ell}}_Y \ket{\nu}_A \\ 
    &= \ket{\psi_*}_X \ket{0^{\ell}}_Y \ket{0^n}_A
\end{align}
as claimed.
\end{proof}

 Now we complete the proof of \cref{thm:anylogn}.  
\begin{proof}[Proof of \cref{thm:anylogn}]\label{proof:here}
Let $\ket{\psi_*}$ be the desired state and let $\ket{\psi'} := \frac{\ket{0^\ell}\ket{0} + \ket{\psi_*} \ket{1}}{\sqrt{2}}$ be an $\ell + 1$ qubit state. 

Recall that for any $k = \polylog(n)$, $\QACZ = \QACZ[\FAN_k]$. 
Combining \cref{cl:smampl} and \cref{cor:parallel} gives a $\poly(n)$ size $\QACZ$ circuit $C'$ for the dirty preparation of $\ket{\psi'}$, i.e, 
such that for $m = \poly(n)$, $C' \ket{0^{\ell+1} } \ket{0^m} = \ket{\psi'} \ket{\nu}$,  for some arbitrary $\ket{\nu}$ on the ancillae. Finally, \cref{cor:clean} provides a $\QACZ$ circuit that cleanly prepares $\ket{\psi_*}$ using $\poly(n)$ ancillae. 
\end{proof}

\section{Acknowledgements} \label{sec:acknowledge}
We credit OpenAI's GPT 6 Pro Astra for providing us the observation leading to \cref{eq:herew}, which was the missing piece in our prior proof attempts. Astra additionally suggested the clean-up strategy we use in \cref{cor:clean}. The proofs presented here are our own versions, written by ourselves, and we take responsibility for their correctness. 
\appendix
\bibliographystyle{alpha}
\bibliography{main}

@inproceedings{rosenthal2021qac0,
  author       = {Gregory Rosenthal},
  title        = {Bounds on the QAC{\^{}}0 Complexity of Approximating Parity},
  booktitle    = {{ITCS}},
  series       = {LIPIcs},
  volume       = {185},
  pages        = {32:1--32:20},
  publisher    = {Schloss Dagstuhl - Leibniz-Zentrum f{\"{u}}r Informatik},
  year         = {2021}
}

@article{moore1999qac0,
  author       = {Cristopher Moore},
  title        = {Quantum Circuits: Fanout, Parity, and Counting},
  journal      = {Electron. Colloquium Comput. Complex.},
  volume       = {{TR99-032}},
  year         = {1999}
}

@article{hoyer2005fanout,
  author       = {Peter Hoyer and
                  Robert Spalek},
  title        = {Quantum Fan-out is Powerful},
  journal      = {Theory Comput.},
  volume       = {1},
  number       = {1},
  pages        = {81--103},
  year         = {2005},
  url          = {https://doi.org/10.4086/toc.2005.v001a005},
  doi          = {10.4086/TOC.2005.V001A005},
  bibsource    = {dblp computer science bibliography, https://dblp.org}
}

@article{takahashi2012collapse,
  author       = {Yasuhiro Takahashi and
                  Seiichiro Tani},
  title        = {Collapse of the Hierarchy of Constant-Depth Exact Quantum Circuits},
  journal      = {Comput. Complex.},
  volume       = {25},
  number       = {4},
  pages        = {849--881},
  year         = {2016}
}

@article{lmn1993ac0,
  author       = {Nathan Linial and
                  Yishay Mansour and
                  Noam Nisan},
  title        = {Constant Depth Circuits, Fourier Transform, and Learnability},
  journal      = {J. {ACM}},
  volume       = {40},
  number       = {3},
  pages        = {607--620},
  year         = {1993}
}

@article{hastad1986switch,
  author       = {Johan Hastad},
  title        = {Almost Optimal Lower Bounds for Small Depth Circuits},
  journal      = {Adv. Comput. Res.},
  volume       = {5},
  pages        = {143--170},
  year         = {1989}
}

@misc{grier2026tc0,
      title={$\mathsf{QAC}^0$ Contains $\mathsf{TC}^0$ (with Many Copies of the Input)}, 
      author={Daniel Grier and Jackson Morris and Kewen Wu},
      year={2026},
      eprint={2601.03243},
      archivePrefix={arXiv},
      primaryClass={cs.CC},
      url={https://arxiv.org/abs/2601.03243}, 
}

@misc{cleandicke2026,
      title={Constant-Depth Unitary Preparation of Dicke States}, 
      author={Malvika Raj Joshi and Francisca Vasconcelos},
      year={2026},
      eprint={2601.10693},
      archivePrefix={arXiv},
      primaryClass={quant-ph},
      url={https://arxiv.org/abs/2601.10693}, 
}

@misc{gretta2026paritynotinqac0iff,
      title={Parity $\notin$ QAC0 $\iff$ QAC0 is Fourier-Concentrated}, 
      author={Lucas Gretta and Meghal Gupta and Malvika Raj Joshi},
      year={2026},
      eprint={2604.02793},
      archivePrefix={arXiv},
      primaryClass={quant-ph},
      url={https://arxiv.org/abs/2604.02793}, 
}

@article{rosenthal2021query,
  author       = {Gregory Rosenthal},
  title        = {Query and Depth Upper Bounds for Quantum Unitaries via Grover Search},
  journal      = {CoRR},
  volume       = {abs/2111.07992},
  year         = {2021},
  url          = {https://arxiv.org/abs/2111.07992},
  eprinttype   = {arXiv},
  eprint       = {2111.07992},
  bibsource    = {dblp computer science bibliography, https://dblp.org},
  note         = {Presented at the 17th Conference on the Theory of Quantum Computation, Communication and Cryptography (TQC 2022)}
}

@misc{gretta2026polylogarithmicweightdickestatesqac0,
      title={Polylogarithmic-Weight Dicke States in QAC$^0$ and Arbitrary Symmetric States in QAC$^0_f$}, 
      author={Lucas Gretta and Meghal Gupta and Malvika Raj Joshi},
      year={2026},
      eprint={2604.15298},
      archivePrefix={arXiv},
      primaryClass={quant-ph},
      url={https://arxiv.org/abs/2604.15298}, 
}

@article{moore2001parallel,
author = {Moore, Cristopher and Nilsson, Martin},
title = {Parallel Quantum Computation and Quantum Codes},
journal = {SIAM Journal on Computing},
volume = {31},
number = {3},
pages = {799-815},
year = {2001},
doi = {10.1137/S0097539799355053},
}

@book{nielsen2000quantum,
  title={Quantum Computation and Quantum Information},
  author={Nielsen, M.A. and Chuang, I.L.},
  isbn={9780521635035},
  lccn={98022029},
  series={Cambridge Series on Information and the Natural Sciences},
  url={https://books.google.com/books?id=65FqEKQOfP8C},
  year={2000},
  publisher={Cambridge University Press}
}

@misc{gottesman1998heisenbergrepresentationquantumcomputers,
      title={The Heisenberg Representation of Quantum Computers}, 
      author={Daniel Gottesman},
      year={1998},
      eprint={quant-ph/9807006},
      archivePrefix={arXiv},
      primaryClass={quant-ph},
      url={https://arxiv.org/abs/quant-ph/9807006}, 
}
\section{Skipped Proofs}\label[appendix]{sec:aproofs}
We prove the remaining corollary below.
\parallelamp
\begin{proof}
Let $t = \lfloor \alpha^{-1} \rfloor$. First, starting with $(m + n) \cdot t$ ancillae, construct $t$ copies of $\ket{\varphi}$ in parallel, each on a $\ell$ qubit target register labeled $Q_i$ and a corresponding ancillae $A_i$ and $v_i$ (for the marked register). Let $V := \clr{v_1, v_2 \dots v_t}$ and $Q_{\all} := Q_1 \cup Q_2 \dots Q_t$ and $A_{\all} := A_1 \cup A_2 \dots A_t$. The resulting state can be written as, 
\begin{align}
    \ket{\psi_0} &:= \bigotimes_{i \in [t]} \lr{\sqrt{\alpha} \ket{\psi_*}_{Q_i} \ket{\mu}_{A_i} \ket{1}_{v_i} + \sqrt{1-\alpha} \ket{\nu}_{Q_i, A_i} \ket{0}_{v_i}} \\
    &= \sqrt{p_*} \lr{\frac{1}{\sqrt{t}} \cdot \sum_{i \in [t]} \ket{\psi_*}_{Q_i} \ket{\mu}_{A_i} \ket{1}_{v_i} \ket{\nu'}_{Q_{\all} \setminus Q_i, A_{\all} \setminus A_i} \ket{0^{t-1}}_{V \setminus v_i}} +  \sqrt{1-p_*} \ket{\bad}.
\end{align}
where $p_* = \Pr[\binomd(t,\alpha) = 1] = \Theta(\Pr[\binomd(t,1/t)]) = \Theta(1)$ and $\ket{\nu'} = \ket{\nu}^{\otimes (t-1)}$. 
Recall that we have the $\EXACT_1$ gate in $\QACZ$ (\cite{cleandicke2026}) using $O(1/\alpha^2)$ ancillae.
Apply this to $V$ with a fresh ancilla to obtain, 
\begin{align}
\ket{\psi_1} &:= \sqrt{p_*} \lr{\frac{1}{\sqrt{t}} \sum_{i \in [t]} \ket{\psi_*}_{Q_i} \ket{\mu}_{A_i} \ket{1}_{v_i} \ket{\nu'}_{Q_{\all} \setminus Q_i, A_{\all} \setminus A_i} \ket{0^{t-1}}_{V \setminus v_i}} \ket{1}_{v_0}  \nonumber \\  & \ \ \ \ \ \  \ + \sqrt{1-p_*} \ket{\bad} \ket{0}_{v_0}.
\end{align}
Then apply amplitude amplification (\cite{grier2026tc0}/Corollary 3.16 of \cite{gretta2026polylogarithmicweightdickestatesqac0}) to obtain in $O(d/p_*)$ depth, the state, 
\begin{align}
\ket{\psi_2} :=  \frac{1}{\sqrt{t}} \sum_{i \in [t]} \ket{\psi_*}_{Q_i} \ket{\mu}_{A_i}  \ket{1}_{v_i} \ket{\nu'}_{Q_{\all} \setminus Q_i, A_{\all} \setminus A_i}\ket{0^{t-1}}_{V \setminus v_i},
\end{align}
where we dropped the cleaned up ancilla $v_0$. 
Now prepare a fresh $\ell$-qubit target register $T$, and apply $t$ $\ctrl{\swap(v_i, Q_i, T)}$ gates, each swapping the registers $Q_i, T$ controlled on $v_i$ to obtain,  
\begin{align}
\ket{\psi_3}_{Q_{\all},A,T} &:= \lr{\prod_{i \in [t]}  \ctrl{\swap(v_i, Q_i, T)}} \cdot \ket{\psi_2} \ket{0^\ell}_T \\
&= \frac{1}{\sqrt{t}} \sum_{i \in [t]} \swap(Q_i,T) \ket{\psi_*}_{Q_i} \ket{0^{\ell}}_T \ket{\mu}_{A_i}  \ket{1}_{v_i} \ket{\nu'}_{Q_{\all} \setminus Q_i, A_{\all} \setminus A_i}\ket{0^{t-1}}_{V \setminus v_i} \\
&= \ket{\eta}_{Q_{\all}, A_{\all}, V} \ket{\psi_*}_T,
\end{align}
for some residual ancilla state $\ket{\eta}$. Due to \cref{lem:ctrswap} this transformation can be implemented in parallel in $O(1)$ depth using $\FANOUT_\ell$. 
\end{proof}

\end{document}